\documentclass[10pt, oneside, reqno]{amsart}

\usepackage{enumitem}
\input ./math_headers_amsart.sty

\bibliography{Twist} 

\def\EE{{\bb E}}

\title{Framed $\bb E_n$-algebras from quantum field theory}
\author{Chris Elliott \and Owen Gwilliam}
\date{\today}

\begin{document}

\begin{abstract}
This paper addresses the following question: given a topological quantum field theory on $\RR^n$ built from an action functional, when is it possible to globalize the theory so that it makes sense on an arbitrary smooth oriented $n$-manifold?  
We study a broad class of topological field theories --- those of AKSZ type --- and obtain an explicit condition for the vanishing of the \emph{framing anomaly}, i.e., the obstruction to performing this globalization procedure.  
We also interpret our results in terms of identifying the observables as an algebra over the framed little $n$-disks operad.
Our analysis uses the BV formalism for perturbative field theory and the notion of factorization homology.
\end{abstract}

\maketitle

\section{Introduction}
Topological field theory has offered a rich domain of common interest for mathematicians and theoretical physicists over the last few decades.
In this paper we examine how and when a constructive method from physics -- the Batalin--Vilkovisky (BV) formalism in conjunction with rigorous renormalization techniques of Axelrod--Singer and Kontsevich for Chern--Simons-type theories -- produces an algebra over the {\em framed} little $n$-disks operad.
Our work here builds upon and extends prior work by the first author \cite{ElliottSafronov} that explains how this constructive method can produce algebras over the framed little $n$-disks operad in general. 
We will see that the obstruction to lifting from the unframed to the framed setting, or {\em framing anomaly}, is always expressed in terms of Pontryagin classes, suitably interpreted.
Our methods are an analog, for a class of theories we will refer to as topological AKSZ theories, of the formalism for anomalies associated with Stora, Wess, and Zumino \cite{Zumino, Stora}, 
but in this topological setting, we can relate directly to the obstruction theory for algebras over these operads.
Here we focus on an explicit computation within the BV framework as articulated by Costello \cite{CostelloBook} and developed further in~\cite{ElliottSafronov}.

Let us describe concisely some concrete consequences of the results proved here.
Our results apply to theories like Chern--Simons theory, topological BF theories, and topological AKSZ theories in general.
Using a simple point-splitting regularization (sometimes called the ``configuration space method''), one can handle divergences in such theories;
the only obstruction to quantization is whether the quantized action satisfies the quantum master equation.
When this obstruction vanishes, the results of \cite{ElliottSafronov} show that the observables of the theory provide an $\EE_n$ algebra.
Here we compute the obstruction-deformation complex describing the ability to lift such an $\EE_n$ algebra structure to a framed $\EE_n$ algebra structure;
we also explain how the obstruction to lifting can be seen as arising from a kind of equivariant quantum master equation.

Why bother to make such a lift? And how do these algebras relate to more conventional approaches to topological field theory?
We will offer answers aimed at topologists and then at physicists.

Functorial field theories, in the style of Atiyah--Segal--Lurie, arise from (framed) $\EE_n$ algebras via factorization homology (see \S4.1 of \cite{LurieTFT} or \cite{ScheimbauerThesis}).
Briefly, an $\EE_n$ algebra $A$ determines a {\em framed} fully extended $n$-dimensional topological field theory with values in a ``higher Morita category'' built from $\EE_n$-algebras. 
A $k$-manifold $X$ with a framing of the bundle $TX \times \RR^{n-k}$ is assigned the invariant 
\[
Z(X) = \int_{X \times \RR^{n-k}} A,
\]
the factorization homology over the $n$-dimensional manifold made by thickening $X$.
The functor $Z$ offers a sophisticated invariant of such $n$-framed manifolds,
but such manifolds are relatively rare. 
(Thank of the $n=2$ case. The only framed closed 2-manifolds are genus 1.)
On the other hand, a framed $\EE_n$ algebra $A$ determines an {\em oriented} fully extended $n$-dimensional topological field theory with values in this ``higher Morita category'' built from $\EE_n$-algebras. 
We now ask that a $k$-manifold $X$ admits an orientation on $TX \times \RR^{n-k}$.
Such manifolds are much more abundant.
Our results thus show how a large class of TFTS -- in the physicist's sense -- determine extended oriented TFTs in the sense of Baez--Dolan and Lurie.

This rather abstract formulation can be expressed in more concrete, physical terms.
The $\EE_n$ algebra of a TFT encodes the operator product expansion of the local operators, with extensive thoroughness.
Think of the local operator that arises from picking a configuration of $k$ distinct, ordered points in $\RR^n$ and inserting a local operator at each point.
Although the value itself is essentially independent of the location of the insertions (you can wiggle the points without changing the output, up to exact terms), 
the topology of the configurations of points is quite rich,
and the $\EE_n$ algebra keeps track of how the OPE depends on that topology.
In other words, it encodes Witten descent and related manipulations.
The associated functorial TFT associates to a $k$-manifold $X$ the $\EE_{n-k}$-algebra encoding the OPE of the full theory dimensionally reduced along~$X$.

Our results explain the conditions under which you can implement this construction -- the OPE algebras and their dimensional reductions -- on oriented manifolds.
In other words, one needs to know how to encode descent given an orientation,
and the anomaly to such descent lies in our obstruction-deformation complex.
In this paper we do not compute any explicit anomalies, leaving that for a forthcoming companion paper \cite{EGWfr}, 
but we do note theories for which the anomaly must vanish because the relevant cohomology group vanishes.  The following is a concrete example.

\begin{example}
Consider a topological BF theory on $\RR^n$ for $n \ge 3$, with gauge Lie algebra $\gg$ a simple Lie algebra.  The results in this paper demonstrate that the only possible framing anomaly for such a theory lies in the Lie algebra cohomology group
\[\bigoplus_{i =1}^{n-1}\mr H^i(\so(n)) \otimes \mr H^{n-i}(\gg).\]
Above degree zero, the cohomology of $\so(n)$ is supported in degrees $3 \text{ mod } 4$, and the cohomology of $\gg$ is supported in odd degrees $\ge 3$.  Thus, we can conclude:
\begin{prop}
For a topological BF theory as above:
\begin{enumerate}
 \item[(1)] The framing anomaly vanishes when the dimension $n$ is odd.
 \item[(2)] For $\gg = \so(k)$ or $\gg = \sp(k)$, the framing anomaly vanishes when the dimension $n$ is not equal to $2 \text{ mod } 4$.
\end{enumerate}
\end{prop}
\end{example}

\begin{example}
Let us now consider the example of 3-dimensional Chern--Simons theory with an arbitrary semisimple gauge group,
which has a well-known framing anomaly for ordinary Chern--Simons theory \cite{WittenJones, AxelrodSinger}.  Although classical Chern--Simons theory can be defined on any oriented 3-manifold, its quantization depends on a choice of framing for the 3-manifold.  The quantization of Chern--Simons theory, including the framing anomaly, is discussed in the language of the BV formalism by Iacovino \cite{Iacovino}. 

\begin{prop}
There is no obstruction to quantizing the $\mf{iso}(3)_\mr{dR}$ action for Chern--Simons theory on $\RR^3$.  However there \emph{is} a potential obstruction to this quantization as an \emph{inner} action.\
\end{prop}
\end{example}

\begin{example}
In higher dimensions, there are \emph{abelian} Chern--Simons theory on $\RR^n$ for any odd integer $n \ge 3$, having to do with connection-type data on higher $U(1)$-gerbes. 
Concretely, we consider the perturbative theories expressed in terms of the formal mapping space $\mr{Map}(\RR^n_{\mr{dR}}, B^{\frac{n-1}2}\mf u(1))$.  
Our methods let us understand the possible obstructions to Chern--Simons theories of this type.

\begin{prop}
There is no obstruction to quantizing the $\mf{iso}(n)_\mr{dR}$ action for Chern--Simons theory on $\RR^n$ with gauge Lie algebra $\mf u(1)$, for any odd integer $n \ge 3$.  However there is a potential obstruction to this quantization as an inner action whenever $n \equiv 3 \text{ mod } 4$.
\end{prop}
\end{example}

\subsection{Overview of the Paper}
We begin in Section \ref{AKSZ_section} by discussing the class of field theory to which our results apply: topological AKSZ theories.  These are topological field theories whose fields can be described in terms of mapping spaces, with BV action functional generated by the AKSZ approach \cite{AKSZ}, via transgression of a shifted symplectic structure on the target of the mapping space.

While these theories make sense on any smooth manifold, in Section \ref{dR_iso_section} we specialize to theories defined on a vector space $\RR^n$, and begin to incorporate the action of the group of isometries.  At the classical level, topological AKSZ theories admit not only an action of the isometry group, but also a trivialization of this action up to homotopy.  The main results of the present paper concern the lift of this homotopy trivialization to the quantum level.  

We discuss the implications of such a lift in Section \ref{En_section}, in which we recall results from \cite{ElliottSafronov} that allow for the realization of a framed $\bb E_n$-algebra structure on the observables of a quantum field theory on $\RR^n$, provided we can define a quantization of the homotopy trivialization of the isometry action.  Such a structure permits the application of the tool of factorization homology to extend such a quantum field theory on $\RR^n$ to more general oriented smooth $n$-manifolds.

In the final section, Section \ref{anomaly_section}, we characterize exactly when it is possible to quantize the homotopically trivial isometry action.  There is a potential anomaly (the framing anomaly) obstructing this quantization, and we explicitly compute the cohomology group in which the obstruction lives.  In many examples, as discussed above, this immediately tells us that the framing anomaly vanishes, so that there is no obstruction to quantization.

\subsection{Acknowledgements}
The authors would like to thank Pavel Safronov and Brian Williams for helpful comments and conversations during the preparation of this paper.  
The National Science Foundation supported O.G. through DMS Grants No. 1812049 and 2042052. 
Any opinions,
findings, and conclusions or recommendations expressed in this material are those of the authors and do not
necessarily reflect the views of the National Science Foundation.

\section{Topological AKSZ Theories} \label{AKSZ_section}

In this paper we will focus on a natural class of topological field theories that can be defined in any dimension, which we will refer to as \emph{topological AKSZ theories}.

\begin{remark}
In this paper we will model classical and quantum field theories in terms of the Batalin--Vilkovisky (BV) formalism \cite{BatalinVilkovisky}.  More specifically, we will be using the model for perturbative classical field theory described in \cite{CostelloBook, Book2}.  See also \cite[Section 1]{ESW} for a summary of the definitions that we will using when we define a classical field theory.  
\end{remark}

\begin{definition}\label{def E_L}
Let $M$ be an oriented $n$-manifold, and let $L$ denote an $L_\infty$ algebra equipped with a cyclic pairing of degree $n-3$.  
We view $L$ as presenting a formal moduli space $\mr B L$ equipped with a shifted symplectic form of degree $n-1$.
The \emph{topological AKSZ theory} on $M$ with target $\mr B L$ is the classical BV theory with whose underlying graded space of fields is
\[
\Omega^\bullet(M) \otimes L[1]
\]
and whose dynamics are encoded by an $L_\infty$ structure on the cochain complex
\[\mc E_L = (\Omega^\bullet(M) \otimes L, \d_{\mr{dR}} \otimes 1 + 1 \otimes \d_L)\]
arising from the wedge product of forms and the brackets on $L$. 
\end{definition}

The pairing on $L$ and integration over $M$ provide a local shifted symplectic structure, or, more accurately, the antibracket on observables (i.e., the Chevalley-Eilenberg cochains of~$\mc E_L$).

We remark that these theories are also often called {\em generalized Chern-Simons theories} \cite{SchwarzGCS, MovshevSchwarz}.

\begin{remark}
If $M$ is compact and $X$ is an $n-1$-shifted symplectic derived stack, then there is a $-1$-shifted symplectic structure on the derived mapping stack
\[\mc M_X(M) = \mr{Map}(M_{\mr{dR}}, X)\]
given by the AKSZ construction of Pantev, To\"en, Vaqui\'e and Vezzosi \cite{PTVV}.  The shifted tangent complex $L = T_x[-1]X$ at a closed point $x$ of $X$ has the structure of an $L_\infty$ algebra with a degree $n-3$ symplectic pairing.  We can identify $\mc E_L$ with the shifted tangent complex of the mapping stack $\mc M$ at the constant map with value $x$. 
\end{remark}

\begin{examples}
Many standard examples fit inside this framework.
\begin{enumerate}
 \item For $n=3$ and $L = \gg$ a reductive Lie algebra equipped with an invariant pairing, the topological AKSZ theory describes perturbative Chern--Simons theory on $M$ with gauge Lie algebra~$\gg$.
 \item For general $n$, let $L = \gg \oplus \gg^*[n-3]$ where $\gg$ is a finite-dimensional Lie algebra acting on $\gg^*[n-3]$ by its coadjoint representation.  In this case the topological AKSZ theory describes perturbative topological BF theory on $M$ with gauge Lie algebra~$\gg$.
 \item More generally, we can replace $\gg$ in the above example by the shifted tangent space $T_y[-1]Y$ to a complex manifold $Y$, and consider 
 \[L = T_y[-1]Y \oplus T^*_y[n-2]Y \iso T_{(y,0)}[-1](T^*[n-1]Y).\]
We can now identify the topological AKSZ theory with the perturbation theory around a constant map of the derived mapping space~$T^*[-1]\mr{Map}(M_{\mr{dR}},Y)$.
\end{enumerate}
\end{examples}

Topological AKSZ theories are extremely amenable to quantization, using techniques developed by Axelrod and Singer \cite{AxelrodSinger} and Kontsevich \cite{KontsevichECM}. 
(See also the summary of Costello, written in language closer to that used in this article~\cite[Section 15]{CostelloBVR}).  We use the term \emph{prequantization} to mean the construction of a family of effective action functionals compatible under the renormalization group flow. 
In this terminology, to provide a \emph{quantization}, these effective action functionals must also satisfy the quantum master equation.

\begin{theorem}
Any topological AKSZ theory can be prequantized to all orders, uniquely up to a contractible choice.  This prequantization can be computed explicitly, and there are no counter-terms.
\end{theorem}

The explicit computation involves a nice description of the propagator, and consequently a computation of the Feynman weights, using partial compactifications of the configurations spaces $\mr{Conf}_m(\RR^n)$ first constructed by Kontsevich \cite{KontsevichECM}.

It will be useful to concretely describe the ring of local functionals associated to the topological AKSZ theory~$\mc E_L$. 
See \cite[Chapter 5, Section 10]{CostelloBook} for more.

\begin{lemma} \label{AKSZ_functionals_lemma}
The ring $\OO_{\mr{loc}}(\mc E_L)$ of local functionals for the theory $\mc E_L$ on an $n$-manifold $M$ is quasi-isomorphic to the shifted de Rham complex $(\Omega^\bullet(M) \otimes \mr C^\bullet_{\mr{red}}(L)[n], \d_{\mr{dR}} \otimes 1 + 1 \otimes \d_{\mr{CE}})$.
\end{lemma}

For $M = \RR^n$, the Poincar\'e lemma then ensures that 
\[
\OO_{\mr{loc}}(\mc E_L) \simeq C^\bullet_{\mr{red}}(L)[n].
\]
In particular, for topological BF theories or Chern-Simons theories for gauge Lie algebra $\gg$, deformations and anomalies correspond to cocycles of Lie algebra cohomology groups for~$\gg$. 
These are well-known for semisimple Lie algebras.

\begin{proof}
See \cite[Lemma 3.5.4.1]{Book2}.
\end{proof}

\section{The de Rham Isometry Action} \label{dR_iso_section}

From now on, let $M = \RR^n$.  We will study anomalies for the action of the isometry group $\ISO(n) = \SO(n) \ltimes \RR^n$ of~$\RR^n$. 
Let $\mf{iso}(n)$ denote the Lie algebra of~$\ISO(n)$.

\begin{definition}
If $\gg$ is a Lie algebra, define $\gg_{\mr{dR}}$ to be the dg Lie algebra whose underlying graded vector space is $\gg[1] \oplus \gg$, with differential given by the identity, and Lie bracket given by the bracket on $\gg$ and the adjoint action of $\gg$ on $\gg[1]$.
\end{definition}

\begin{remark}
This dg Lie algebra $\gg_{\mr{dR}}$ is homotopy equivalent to a trivial Lie algebra.  On the other hand, it has an important interpretation from the point of moduli spaces:
its associated formal moduli space offers a useful model of the {\em de Rham space} $\mr B\gg_{\mr{dR}}$ of the formal moduli space $\mr B\gg$.
In more explicit terms, note that there is a natural map of dg Lie algebras $\gg \to \gg_{\mr{dR}}$. 
A representation of $\gg_{\mr{dR}}$ pulls back to a representation of $\gg$,
but with an explicit trivialization (up to chain homotopy).
Indeed, we can view the representations of $\gg_{\mr{dR}}$ as the representations of $\gg$ equipped with a homotopical trivialization.
\end{remark}

Every topological AKSZ theory on $\RR^n$ has a natural action of $\mf{iso}(n)$ by the Lie derivative action of vector fields, since this Lie algebra acts canonically on the de Rham complex.
This action extends canonically to an action by $\mf{iso}(n)_{\mr{dR}}$,
where the component $\mf{iso}(n)[1]$ acts by contraction of vector fields with differential forms, thanks to Cartan's formula.  
This action can be encoded by a current, in the sense of Noether, as follows.
Consider the degree -1 local functional
\[S_{\mr{eq}} \in \mr C^\bullet(\mf{iso}(n)_{\mr{dR}}, \OO_{\mr{loc}}(\mc E_L))\]
defined by the formula
\[S_{\mr{eq}}(\wt X, X)(A) = S(A) - \int \left(\langle A \wedge \iota_{\wt X} A \rangle + \langle A \wedge L_{X} A \rangle \right).\] 
Here $(\wt X, X)$ is an element of $\mf{iso}(n)_{\mr{dR}}$, $A$ is an element of $\mc E_L$, and $\langle-,- \rangle$ denotes the symplectic pairing on $L$.
This current determines a derivation $\{S_{\mr{eq}}, -\}$ acting on the classical observables, as we now verify.

\begin{prop} \label{isometry_action_prop}
There is a map of dg Lie algebras from $\mf{iso}(n)_{\mr{dR}}$ to vector fields on the formal moduli space $\mc E_L$ given by $\{S_{\mr{eq}},-\}$. 
In particular, this local functional $S_{\mr{eq}}$ satisfies the equivariant classical master equation.
\end{prop}

\begin{proof}
The functional $S$ satisfies the classical master equation by assumption, so we only need to consider terms in the equivariant classical master equation with a non-trivial dependence on the auxiliary (or background) fields $X$ or $\wt X$ from~$\mf{iso}(n)_{\mr{dR}}$.  

Write $I_X(A) = -\langle A \wedge L_{X} A \rangle$, and $J_{\wt X}(A) = -\langle A \wedge \iota_{\wt X} A \rangle$.  
We will show that the following equations hold: 
\begin{align}
\{S, I_X\} &= 0 \\
\frac 12 \{I_X, I_X\} + \d_{\mr{CE}}I_{X} &= 0 \\
\{S, J_{\wt X}\} + \{I_X, J_{\wt X}\} + \d_{\mr{CE}}J_{\wt X} &= 0 \\
\frac 12 \{J_{\wt X}, J_{\wt X}\} &= 0.
\end{align}
Equations (1) and (2) together say that $\mc E_L$ is $\mf{iso}(n)$-equivariant, which follows by observing that there is a smooth action of the \emph{Lie group} $\mr{ISO}(n)$ on $\mc E_L$ by isometries of $\RR^n$, which is infinitesimally generated by the functional~$I_X$.

Equation (4) is straightforward: it follows from the fact that $\iota_{\wt X}^2 = 0$.  It remains to deduce equation~(3), which is a consequence of Cartan's formula.  Indeed, 
\begin{align*}
\d_{\mr{CE}}J_{\wt X}(A) &= -\left\langle A \wedge (L_{\wt X}(A) + \iota_{[X,\wt X]}(A)) \right\rangle \\
&= -\left\langle A \wedge ([\d, \iota_{\wt X}](A) + [L_X,\iota_{\wt X}](A)) \right\rangle \\
&= -\{S, J_{\wt X}\}(A) - \{I_X, J_{\wt X}\}(A).
\end{align*}

\end{proof}

Let us restrict this action to an action of the ordinary algebra of isometries $\mf{iso}(n)$ alone, acting by the Lie derivative.  
This action can be defined at the quantum level, and it naturally comes from a smooth action of the \emph{group}~$\mr{ISO}(n)$.  

\begin{prop} \label{quantum_iso_action_prop}
There is a smooth classical action of the Lie group $\mr{ISO}(n)$ on the topological AKSZ theory~$\mc E_L$.  This action can be lifted to an action at the quantum level. 
\end{prop}

We applied this result in a specific family of examples in \cite[Proposition 5.10]{EGWHigherDef}, using the same argument.

\begin{proof}
This claim follows from the result \cite[Proposition 9.1.1.2]{Book2}.  This proposition proves the given claim for the group of translations, but as remarked following the result in {\it loc. cit.}, the same argument works for the full group of isometries.  According to the cited result, it suffices to prove that the effective interaction associated to any parametrix is isometry invariant.  In turn, it is enough for the classical interaction, along with the choice $\d^*$ of gauge-fixing operator to be isometry invariant.
\end{proof}

Likewise, let us restrict the $\mf{iso}(n)_{\mr{dR}}$ action to an action of $\RR^n_{\mr{dR}}$.  We can, again, define this action at the quantum level.

\begin{prop} \label{translation_dR_no_obstruction_prop}
There is no anomaly obstructing the lift of the classical action of $\RR^n_{\mr{dR}}$ on the topological AKSZ theory $\mc E_L$ to the quantum level. 
\end{prop}

\begin{proof}
We will prove this claim by thinking about the weights of Feynman diagrams that would generate an anomaly obstructing the quantization of such an action.  Consider a Feynman diagram of shape $\Gamma$ containing a vertex at position $x \in \RR^n$ labelled by the interaction $J_{\wt X}(A)$, where $\wt X$ is the vector field generating a translation.

The weight of the diagram $\Gamma$ can be decomposed as a sum of weights $W_{\Gamma, e}$, where a single internal edge $e$ of $\Gamma$ is labelled by the heat kernel, and the remaining edges are all labelled by the propagator.  Let us show that this sum vanishes.  There are two classes of summand:

 \begin{enumerate}
 \item Suppose we label $\Gamma$ so that the special edge $e$ labelled by $K$ is not adjacent to the vertex at $x$. Then the associated Feynman weight is a limit of terms of the form
 \[\int_{t \in [\eps, \Lambda]} \int_{(x_1, \ldots, x_{N-1}) \in (\RR^n)^{N-1}} \left(\int_{x \in \RR^n} \d^*K_t(x-x_1) \wedge \iota_{\dd_j} \d^*K_t(x-x_2)\right) \wedge F(x_1, \ldots, x_{N-1}),\]
 where $F$ is some differential form (we won't need its explicit form, only the fact that it is independent of the location $x$).  Because $\d^*$ and $\iota_{\dd_j}$ commute, the term inside the parentheses vanishes, so $W_{\Gamma, e}=0$.  
 \item There are two labellings where the special edge $e$ labelled by $K$ is adjacent to $x$, say $e=e_1$ or $e=e_2$.  The weights of these two labelled diagrams differ by a sign, and therefore they cancel when we sum over all labellings.  Indeed, by integration by parts in the $x$ variable
 \begin{align*}
W_{\Gamma, e_1} &= \int_{t \in [\eps, \Lambda]} \int_{(x_1, \ldots, x_{N-1}) \in (\RR^n)^{N-1}} \left(\int_{x \in \RR^n} \d^*K_t(x-x_1) \wedge \iota_{\dd_j} K_t(x-x_2)\right) \wedge F(x_1, \ldots, x_{N-1}) \\
& = - \int_{t \in [\eps, \Lambda]} \int_{(x_1, \ldots, x_{N-1}) \in (\RR^n)^{N-1}} \left(\int_{x \in \RR^n} K_t(x-x_1) \wedge \iota_{\dd_j} \d^*K_t(x-x_2)\right) \wedge F(x_1, \ldots, x_{N-1}) \\
&= - W_{\Gamma, e_2}.
 \end{align*}
 \end{enumerate}
 
As a result, the sum of the weights $W_{\Gamma, e}$ over all edges vanishes, which implies that the anomaly for the $\RR^n_{\mr{dR}}$ action vanishes as claimed.
\end{proof}

So, putting this together, we find it is always possible to quantize a topological AKSZ theory on $\RR^n$ equivariantly for the action of the dg Lie algebra $\so(n) \ltimes (\RR^n_{\mr{dR}})$.  In the next section we will fix an equivariant quantization for this dg Lie algebra, and study lifts to $\mr{iso}(n)_{\mr{dR}}$-equivariant quantizations.

\section{\texorpdfstring{$\EE_n$}{En}-Algebras from Topological AKSZ theories} 
\label{En_section}

Let us briefly review the relationship between topological field theories and $\bb E_n$ algebras as described in \cite{ElliottSafronov}.  Consider a classical field theory $\mc E$ on $\RR^n$, and suppose that $\mc E$ admits a smooth action of $\RR^n_{\mr{dR}}$ as discussed in the previous section.  For example, $\mc E$ might be a topological AKSZ theory.  We can often describe either the classical or the quantum \emph{observables} of the field theory using the language of homotopical algebra.

Recall that an \emph{$\bb E_n$-algebra} is defined as a module, in the category of cochain complexes, over the operad of little $n$-disks.  A \emph{framed} $\bb E_n$-algebra is an $\bb E_n$-algebra equipped with a compatible action of the group $\SO(n)$ of rotations.  In this section we will discuss the realization of $\bb E_n$-algebras as a special case of the theory of factorization algebras, as developed in \cite{Book1, Book2} in the context of quantum field theory. 

Let us write $\obscl(\mc E)$ for the factorization algebra of classical observables of the theory $\mc E$.  This factorization algebra inherits a smooth action of $\RR^n_{\mr{dR}}$ from the action on the classical fields.  If, furthermore, there is no anomaly obstructing the action of $\RR^n_{\mr{dR}}$ at the quantum level -- for instance, for topological AKSZ theories by Proposition \ref{translation_dR_no_obstruction_prop} -- then there an action of $\RR^n_{\mr{dR}}$ on the factorization algebra $\obsq(\mc E)$ of \emph{quantum} observables.  This is exactly the context in which we can invoke the following result.

\begin{definition}
Let $\obs$ be a factorization algebra on $\RR^n$ with a smooth action of $\RR^n_{\mr{dR}}$.
It is \emph{rescaling-invariant} if the structure map 
\[
\obs(B_r(0)) \to \obs(B_R(0))
\] 
for the inclusion of concentric balls is a quasi-isomorphism for any $r < R$.
\end{definition}

\begin{theorem}[{\cite[Corollary 2.30]{ElliottSafronov}}]
Let $\obs$ be a rescaling-invariant factorization algebra on $\RR^n$ with a smooth action of $\RR^n_{\mr{dR}}$.
Then the cochain complex $\obs(B_1(0))$ of observables on the unit ball can be canonically equipped with the structure of an $\bb E_n$ algebra.
\end{theorem}

\begin{remark}
We will apply this result to the factorization algebra of quantum observables of a topological AKSZ theory, where the condition of rescaling invariance is automatically satisfied.  
At the classical level it is immediate from Lemma \ref{AKSZ_functionals_lemma}, since the de Rham complex is locally constant.  
When we quantize, as a graded vector space the quantum observables are isomorphic to 
\[\Omega^\bullet(U) \otimes \mr C_{\mr{red}}^\bullet(L)[n])[\![\hbar]\!],\]
and we need to observe that the quantum corrections to the differential on the factorization algebra of observables do not violate rescaling invariance.  We can see this using the spectral sequence associated to the filtration by $\hbar$ degree, 
whose $E_2$ page recovers the factorization algebra of classical observables.  
The rescaling map is a map of filtered complexes and induces a quasi-isomorphism on the $E_2$ page of this spectral sequence, so we therefore obtain a quasi-isomorphism at the $E_\infty$ page.
\end{remark}

If we can promote the smooth action of translations to a smooth action of rotations, then we can strengthen this result to provide a \emph{framed} $\bb E_n$ algebra structure.

\begin{theorem}[{\cite[Corollary 2.39]{ElliottSafronov}}]
Let $\obs$ be a rescaling-invariant factorization algebra on $\RR^n$ with a smooth action of $\mr{ISO}(n)_{\mr{dR}}$.
Then the cochain complex $\obs(B_1(0))$ of observables on the unit ball can be canonically equipped with the structure of an $\bb E_n^{\mr{fr}}$ algebra.
\end{theorem}

Field theories provide our main source of factorization algebras, 
by the central result of \cite{Book2}:
a BV theory on $\RR^n$ determines a factorization algebra on $\RR^n$.
Hence a deformation of the theory determines a deformation of the factorization algebra, and
in fact there is a map from the deformation complex of the theory to the deformation complex of its factorization algebra of observables.
For this reason, if we want to show that a group acts smoothly on the observables,
it suffices to understand how it acts on the theory.
In particular, for this paper, we want to characterize when a quantization is $\mr{ISO}(n)_{\mr{dR}}$-equivariant. 

For a topological theory, as given by Definition~\ref{def E_L}, we have seen that the theory (and its usual quantization) is rescaling-invariant and $\mr{ISO}(n)$-equivariant,
and thus so is the factorization algebra of observables.
In fact, we have also shown that the translation action is homotopically trivial,
so what remains is to trivialize homotopically the $\so(n)$-action.

\section{Computation of Framing Anomalies} \label{anomaly_section}

Let's follow the procedure we just outlined in Section \ref{dR_iso_section}, using the classical action of $\mf{iso}(n)_\mr{dR}$ of Proposition \ref{isometry_action_prop} and the quantization of the algebras $\mf{iso}(n)$  and $\RR^n_{\mr{dR}}$  that we have already constructed in Proposition \ref{quantum_iso_action_prop} and Proposition \ref{translation_dR_no_obstruction_prop}, respectively.  
We would like to lift this to an action of all of $\mf{iso}(n)_{\mr{dR}}$ at the quantum level. Our main result is the following.

\begin{theorem} \label{obstruction_theorem}
Fix an $\so(n) \ltimes (\RR^n_{\mr{dR}})$-equivariant quantization of the topological AKSZ theory $\mc E_L$ associated to a cyclic $L_\infty$ algebra $L$, 
as described in Definition~\ref{def E_L}.  
The obstruction to lifting this quantization to an $\mr{iso}(n)_{\mr{dR}}$-equivariant quantization is given by an element in
\begin{equation}
\label{anomaly_cohomology_eqn}
\bigoplus_{i+j=n}\mr H^i_{\mr{red}}(\so(n)) \otimes \mr H^j_{\mr{red}}(L).
\end{equation}
The obstruction to lifting to an \emph{inner} $\mr{iso}(n)_{\mr{dR}}$-equivariant quantization is given by an element~of
\begin{equation}
\bigoplus_{i+j=n}\mr H^i_{\mr{red}}(\so(n)) \otimes \mr H^j(L).
\end{equation}
\end{theorem}

\begin{corollary}
If the cohomology group \ref{anomaly_cohomology_eqn} vanishes then the factorization algebra of quantum observables for the topological AKSZ theory $\mc E_L$ can be canonically equipped with the structure of a framed $\bb E_n$-algebra.
\end{corollary}

We will outline the argument and then prove the intermediate results that realize it.
The reader should pay attention to how Pontryagin classes can be seen as labeling obstruction classes.

First, we identify the obstruction-deformation complex where the obstruction to our quantization will live.  
Let \[\mr{Act}_{\gg}(\mc E_L)= \mr C^\bullet_{\mr{red}}(\gg, \OO_{\mr{loc}}(\mc E_L))\] denote the formal moduli space describing $\gg$-equivariant deformations of a classical theory $\mc E_L$. 
(For an overview, see Chapter 11 of \cite{Book2}, and for extensive discussion, see Section 2, Chapter 12 and Section 2, Chapter 13 of \cite{Book2}.)
Its tangent complex is a cochain complex, and
the obstruction to $\gg$-equivariant quantization is a degree 1 cocycle in that complex. 
These results lead to equivariant refinements of Lemma~\ref{AKSZ_functionals_lemma},
which characterize the equivariant local functionals up to equivalence:
\[
\mr{Act}_{\gg}(\mc E_L)  \simeq \mr C^\bullet_{\mr{red}}(\gg, C^\bullet_{\mr{red}}(L)[n]).
\]
In this paper, $\gg$ will be $\mf{iso}(n)_{\mr{dR}}$ or~$\so(n) \ltimes \RR^n_{\mr{dR}}$. 

By hypothesis, we have an $\so(n) \ltimes \RR^n_{\mr{dR}}$-equivariant quantization and we are asking to lift to an $\mf{iso}(n)_{\mr{dR}}$-equivariant quantization.
Hence we need to describe the fiber of the map
\[
\mr{Act}_{\mf{iso}(n)_{\mr{dR}}} \to \mr{Act}_{\so(n) \ltimes \RR^n_{\mr{dR}}}
\]
to characterize the lifting problem.
This fiber is derived in nature;
it has an explicit cochain model $\mc C_{n,L}$ that we describe in Lemma~\ref{description_of_fiber_lemma} below.

\begin{remark}
The case of an inner action is quite similar. 
Here we extend $\mr{Act}_{\gg}(\mc E_L)$ to $\mr{InnerAct}_{\gg}(\mc E_L)$ by $\mr C^\bullet_{\mr{red}}$. 
(See Lemma 12.2.3.2 of \cite{Book2}.)
Concretely, we are asking for the $\gg$-action to be inner, i.e., realized by local functionals.
\end{remark}
 
Next, we begin to calculate the cohomology of $\mc C_{n,L}$ by using a  spectral sequence arising from a natural filtration on the complex $\mc C_{n, L}$. 

The first pages of this spectral sequence reduce to computations well-known from topology.
Recall that the cohomology of the classifying spaces $\mr{BSO}(n)$ are graded polynomial rings with even generators given by the Pontryagin classes (and when $n$ is even, an extra generator called the Pfaffian).
Explicitly,
\begin{equation}\label{hbso}
\mr H^\bullet(\mr{BSO}(n), \CC) \iso 
\begin{cases} 
\CC[p_1, p_2, \ldots, p_{k}] & \text{$n =2k+1 $ odd} \\
\CC[p_1, p_2, \ldots, p_{k-1}, p'_k] & \text{$n = 2k$ even} \\
\end{cases}
\end{equation}
where each generator $p_j$ has degree $4j$ and, for $n=2k$ even, the generator $p'_k$  has degree $2k=n$.
Recall as well that the Lie algebra cohomology $\mr H^\bullet(\so(n))$ equals the cohomology $\mr H^\bullet(\SO(n),\CC)$;
these are graded polynomial rings with odd generators, where each generator's degree is one less than the corresponding Pontryagin class.
Explicitly,
\begin{equation}\label{hso}
\mr H^\bullet(\so(n)) \iso 
\begin{cases} 
\CC[\eta_1, \eta_2, \ldots, \eta_{k}] & \text{$n =2k+1 $ odd} \\
\CC[\eta_1, \eta_1, \ldots, \eta_{k-1}, \eta_{k}'] & \text{$n = 2k$ even} \\
\end{cases}
\end{equation}
where $|\eta_k| = 4j-1$ and, for $n=2k$ even, $|\eta'_k| = |p'_k|-1$.
Thus, there are generators in degrees 3, 7, and so on.

In terms of those cohomology rings, the $E_2$-page of the spectral sequence computing $\mr H^\bullet(\mc C_{n,L})$ is 
\[
E_2^i \iso\bigoplus_{\substack{j+k+\ell = n+i \\ k>0}} \mr H^j(\so(n)) \otimes \mr H^{k}(\mr{BSO}(n)) \otimes \mr H^\ell_{\mr{red}}(L).
\]
This isomorphism is the content of Proposition~\ref{E2_page_prop}.  
The differential on this $E_2$-page sends each $\eta_j$ to $p_j$.  
As shown in Lemma~\ref{d3_differential_lemma}, the $E_3$-page is then isomorphic to
\[
\mr H^i(\mc C_{n,L}) 
= \bigoplus_{\substack{j+\ell=n+i-1\\ i>0}} \mr H^j(\so(n)) \otimes \mr H^\ell_{\mr{red}}(L)
\]
and the spectral sequence collapses on this page.
Thus, we know that anomalies obstructing the $\so(n)_{\mr{dR}}$ action live in
\[
\bigoplus_{\substack{j+\ell=n\\j>0}} \mr H^j(\so(n)) \otimes \mr H^\ell_{(\mr{red})}(L),
\]
as claimed.

\begin{remark}
We have identified the space of possible anomalies abstractly in terms of classes in $\mr H^{>0}(\so(n))$, but we can make our description more explicit.  That is, we can describe where these classes came from in $\mr H^{>0}(\mr{BSO}(n)) \otimes \mr H^\bullet(\so(n))$.  The classes that survive to the $E_\infty$ page of the spectral sequence are all linear in the Pontryagin classes $p_k$.  If we identify $\mr H^\bullet(\mr{BSO}(n)) \otimes \mr H^\bullet(\so(n))$ as the polynomial algebra in the classes $p_j$ and $\eta_j$, we can identify the factor in our spectral sequence surviving to the $E_\infty$ page as the image of the generators of $\RR[\eta_1, \ldots, \eta_k]$ under the differential induced by the map sending $\eta_i$ to $p_i$ for all $i$.  In other words, the classes that survive take the form
\[\sum_{j=1}^\ell \eta_{i_1} \eta_{i_2} \cdots \eta_{i_{j-1}} p_{i_j} \eta_{i_{j+1}} \cdots \eta_{i_\ell},\]
for any sequence $1 \le i_1 < i_2 < \cdots < i_\ell \le k$.

This description will follow directly from our proof of Theorem~\ref{obstruction_theorem}.  We discuss the simplest two examples, where the dimension $n$ is equal to 3 or 4, in Examples \ref{3d_example} and \ref{4d_example}. 
\end{remark}

Now that we have traced the path,
we will begin with the first step.

\begin{lemma} \label{description_of_fiber_lemma}
The fiber of the map 
\begin{equation}\label{htpyfib}
\mr{Act}_{\mf{iso}(n)_{\mr{dR}}} \to \mr{Act}_{\so(n) \ltimes \RR^n_{\mr{dR}}}
\end{equation}
is quasi-isomorphic to $\mc C_{n,L}$, whose underlying graded vector space agrees with that of
\[
\mr C^\bullet(\so(n), \sym^{>0}(\so(n)^*[-2]) \otimes \Omega^\bullet(\RR^n)) \otimes \mr C^\bullet_\mr{red}(L)[n],
\]
but whose differential is 
\begin{equation} \label{fib_dif}
(\d_{\mr{CE}}+ \d_{\mr{dR}} + \d') \otimes 1 + 1 \otimes \d_{\mr{CE}},
\end{equation}
where $\d_{\mr{CE}}$ is the Chevalley--Eilenberg differential (for the relevant Lie algebra acting on the relevant module), $\d_{\mr{dR}}$ is the de Rham differential on $\Omega^\bullet(\RR^n)$, and $\d'$ is the operator extended as a derivation from the identity map $\so(n)^*[-1] \to \so(n)^*[-2]$.
\end{lemma}

Observe that this model of the fiber is, in fact, a dg commutative algebra: the tensor factors are dg commutative algebras and the differential can be checked to be a derivation.

\begin{proof}
We can describe the complexes $\mr{Act}_{\mf{iso}(n)_{\mr{dR}}}$ and  $\mr{Act}_{\so(n) \ltimes \RR^n_{\mr{dR}}}$ directly as in \cite[Section 11.2]{Book2}.  This complex $\mc C_{n,L}$ is the set-theoretic fiber product (i.e., kernel of the map of cochain complexes).
But, using the projective model structure \cite{HinichHAHA} on cochain complexes (or dg Lie algebras), we see the map \eqref{htpyfib} is a fibration and so the kernel provides the homotopy fiber product.
\end{proof}

As discussed in the outline, we now consider the spectral sequence associated to the filtration that turns on the term $\d'$ in the differential.  

\begin{prop}\label{E2_page_prop}
Consider the spectral sequence associated to the filtration $F_p$ on the complex $\mc C_{n,L}$ with
\[
F_p\, \mc C_{n,L} =  
\bigoplus_{a \ge p}\mr C^\bullet(\so(n), {\sym}^{a}(\so(n)^*[-2]) \otimes \Omega^b(\RR^n)) \otimes \mr C^\bullet_\mr{red}(L)[n]
\]
for $p \ge 1$, where the right hand side is equipped with the differential~\eqref{fib_dif}, and $F_{\le 0} \, \mc C_{n,L} = 0$.

The $E_2$-page of this spectral sequence is equivalent to 
\begin{align*}
E_2^i &\iso \bigoplus_{\substack{j+k+\ell = n+i \\ k>0}} \mr H^j(\so(n)) \otimes \mr H^{k}(\mr{BSO}(n)) \otimes \mr H^\ell_{\mr{red}}(L).
\end{align*}
\end{prop}

This filtration produces a spectral sequence of graded commutative algebras.
This isomorphism on the $E_2$-page is, in fact, a map of graded commutative algebras.

To compute the $E_2$-page of this spectral sequence, the following result is useful.

\begin{lemma}\label{invariant_coeff_prop}
If $V$ is a finite-dimensional $\so(n)$-representation, then there is a natural isomorphism
\[\mr H^\bullet(\so(n), V \otimes \Omega^\bullet(\RR^n)) \iso \mr H^\bullet(\so(n)) \otimes (V \otimes \Omega^\bullet(\RR^n))^{\so(n)}.\]
\end{lemma}

In words, computing the cohomology decouples into knowing $H^\bullet(\so(n))$ and knowing the strict invariants of $V$-valued differential forms.

\begin{proof}[Proof of Lemma~\ref{invariant_coeff_prop}]
The complex $\mr C^\bullet(\so(n), V \otimes \Omega^\bullet(\RR^n))$ is the totalization of a double complex where one differential is the exterior derivative and the other is the Chevalley-Eilenberg differential.
Consider the spectral sequence of this double complex, where we take the exterior derivative first.
Then the $E_2$-page is $\mr H^\bullet(\so(n), V)$.
For any finite-dimensional representation $W$ of a semisimple Lie algebra $\gg$, there is a natural isomorphism
\[
H^\bullet(\gg, W) \iso H^\bullet(\gg) \otimes W^\gg
\]
so the $E_2$-page is isomorphic to $\mr H^\bullet(\so(n)) \otimes V^{\so(n)}$.
The sequence collapses on this page, so the claim is shown.
\end{proof}

That lemma makes the proof of the proposition straightforward.

\begin{proof}[Proof of Proposition~\ref{E2_page_prop}]
By examining the filtration, one finds that computing the $E_2$-page boils down to computing the cohomology of the double complex
\[
\mr C^\bullet(\so(n), {\sym}^{a}(\so(n)^*[-2]) \otimes \Omega^\bullet(\RR^n))
\]
for each natural number~$a$.
But then
\[
\left( {\sym}(\so(n)^*[-2]) \right)^{\so(n)} \iso \mr H^\bullet(\mr{BSO}(n), \CC),
\]
by Chern--Weil theory, as in~\cite{Chern}.
\end{proof}

\begin{lemma}\label{d3_differential_lemma}
The differential $\d_3$ on the $E_2$-page
\[
E_2^i \iso
\bigoplus_{\substack{j+k+\ell = n+i\\k>0}} \mr H^j(\so(n)) \otimes \mr H^{k}(\mr{BSO}(n)) \otimes \mr H^\ell_{\mr{red}}(L)
\]
of our spectral sequence is induced, as a module over $\mr H^\bullet(\mr{BSO}(n))$, by the map sending $\eta_j \otimes \alpha$ to $1 \otimes (p_j \wedge \alpha)$, where $1 \otimes (p_j \wedge \alpha) \in \mr H^\bullet(\so(n)) \otimes \mr H^\bullet(\mr{BSO}(n))$.
\end{lemma}

\begin{proof}
This statement follows immediately from the definition of the spectral sequence of a filtered complex.  The differential on the $E_2$ page is inherited from the restriction of the differential \eqref{fib_dif} to the space of 2-almost cycles, i.e. those terms closed for the piece of the differential that does not raise filtered degree.  This restricted differential is identical to the restriction of the summand $\d' \otimes 1$ of the differential, and $\d'$ acts on generators of $\sym^\bullet(\so(n)^*)$ exactly as stated.
\end{proof}

\begin{example} \label{3d_example}
It may be useful to the reader to understand the cohomology of the differential $\d_3$ in some small examples.  Let us consider the example where $n=3$, and where $L$ is trivial.  So $\mr H^\bullet(\mr{BSO}(3)) \iso \RR[p]$  is a polynomial ring in a single variable of degree 4, and $\mr H^\bullet(\so(3)) \iso \RR[\eta]$ is an exterior algebra in a single variable of degree 3.  Our $E_2$ page is therefore identified with the ideal $I$ in the ring $\RR[\eta, p]$ generated by $p$, and the differential $\d_3$ sends the generator $\eta$ to $p$.  In terms of a linear basis we can illustrate the complex $(E_2, \d_3)$ pictorially as
\[\xymatrix{
p & p^2 & p^3 & p^4 & \cdots \\
p\eta \ar[ur] &p^2\eta \ar[ur] &p^3\eta \ar[ur] &\cdots &
}\]
where the arrows all represent isomorphisms between one-dimensional summands.  So, when we compute the cohomology with respect to $\d_3$ the result is the one-dimensional vector space generated by $p$.
\end{example}

\begin{example} \label{4d_example}
If we consider the next simplest example, where $n=4$, we now have a pair of even generators for $\mr H^\bullet(\mr{BSO}(4))$, namely the first Pontryagin class $p$ and the Pfaffian $p'$, 
and we have a corresponding pair of odd generators $\eta, \eta'$ for $\mr H^\bullet(\so(4))$.  When we compute the cohomology with respect to the differential $\d_3$, we find a three-dimensional vector space, spanned by the classes $p$, $p'$, and $p\eta' - p'\eta$.
\end{example}

Now we are finally ready to prove the main theorem of this section.

\begin{proof}[Proof of Theorem \ref{obstruction_theorem}]
It is sufficient to identify the $E_3$-page of our spectral sequence with the desired expression, 
and to verify that the higher differentials all vanish.

First, we reinterpret the $E_2$-page in more algebraic terms.  
Observe that the graded vector space 
\[H^\bullet(\so(n)) \otimes \mr H^\bullet(\mr{BSO}(n)) \otimes \mr H^\ell_{\mr{red}}(L)\]
naturally forms a graded commutative algebra, 
as it is the tensor product of three graded commutative algebras.
Let $A = \mr H^\ell_{\mr{red}}(L)$ denote the nonunital algebra of the third tensor factor.
Then we can write the full algebra as
\[R = A[\{\eta_j\}],\{p_j\}]\]
the free commutative algebra over $A$ with the generators $c_j$ and $\eta_j$ from the algebras $\mr H^\bullet(\so(n))$ and $\mr H^\bullet(\mr{BSO}(n))$.
More useful for us is this algebra's maximal ideal 
\[\overline{R} = \bigoplus_{m+n > 0} A\, \eta^m p^n ,\]
spanned by monomials other than the unit monomial $c^0 \eta^0$.
(Here $m$ and $n$ are vectors that encode multiple exponents. 
For instance, $m = (m_1, m_2,\ldots)$ and $\eta^m = \eta_1^{m_1} \eta_2^{m_2} \cdots$.)
But the $E_2$-page corresponds to the subcomplex 
\[\mr H^\bullet(\so(n)) \otimes \mr H^\bullet_{\mr{red}}(\mr{BSO}(n)) \otimes \mr H^\ell_{\mr{red}}(L),\]
which is the ideal $I \sub \overline{R}$ generated by~$(\{p_j\})$.
Note that the quotient algebra is 
\[\overline{R}/I =\bigoplus_{m> 0} A \,\eta^m,\] 
the maximal ideal of the polynomial ring over $A$ generated by all the~$\eta$s.
In fact, $\overline{R}/I \iso \mr A \otimes \mr H_{\rm{red}}(\so(n))$.

Now we turn to the differential, which has a convenient description in terms of these algebraic structures.
The filtration from Proposition~\ref{E2_page_prop} has an analog where we work with all of $\sym(\so(n)[-2])$ and do not keep only positive symmetric powers.
Its $E_2$-page can be identified with $R$.
The differential on that page makes $R$ into a dg algebra over $A$ with derivation $\d$ sending $\eta_j$ to $p_j$.
This differential makes both $\overline{R}$ and $I$ into {\em dg} ideals; 
and this dg ideal $(I, \d)$ is precisely the $E_2$-page of our spectral sequence, as can be verified by unwinding the construction.

Hence, the $E_3$-page is the cohomology of this dg ideal $(I, \d)$.  
To compute its cohomology, we use the long exact sequence associated to the short exact sequence
\[0 \to I \to \overline{R} \to \overline{R}/I \to 0,\]
where we mean the cochain complexes with differential $\d$.
For the middle term, observe that $\mr H^\bullet(\overline{R},\d) = 0$ by direct computation (as any monomial goes to another monomial).
For $\overline{R}/I$, the differential inherited from $\overline{R}$ is zero, so the $(i+1)^{\text{st}}$ cohomology group is simply the degree $i+1$ component $(\overline{R}/I)^{i+1}$ of $\overline{R}/I$.
Hence, the long exact sequence tells us that $\mr H^0(I) = 0$ and 
\[ \mr H^i(I) \iso (\overline{R}/I)^{i+1}\]
for all $i > 0$. 

We conclude that
\[
 E_3^i \iso (R/I)^{i+1} = \bigoplus_{j+k=n+i}\mr H^j_{\mr{red}}(\so(n)) \otimes \mr H^k_{\mr{red}}(L).
\]
These cohomology classes are all represented by elements in the $E_2$-page $\mr H^\bullet(\so(n)) \otimes \mr H^\bullet(\mr{BSO}(n)) \otimes \mr H^\bullet_{\mr{red}}(L)$ with degree zero in the $\mr H^\bullet(\mr{BSO}(n))$ factor. 
In other words, these are polynomials in the generators $p_j$, not involving the $\eta_i$.  
There are no higher differentials in our spectral sequence between terms of this type, 
and so the spectral sequence collapses at the $E_3$-page.  We have, therefore, obtained the equivalence we desired.
 
The case of inner actions remains, but the proof carries over easily.
We need to replace the (homotopy) fiber of the map \eqref{htpyfib} by the (homotopy) fiber of the analogous map.  That is, the fiber of the map 
\begin{equation}
\mr{InnerAct}_{\mf{iso}(n)_{\mr{dR}}} \to \mr{InnerAct}_{\so(n) \ltimes \RR^n_{\mr{dR}}}
\end{equation}
where $\mr{InnerAct}$ means we allow local functionals purely of the background fields (as in the non-inner case, we refer to \cite[Section 11.2]{Book2}).
Concretely, that means the fiber is quasi-isomorphic to $\mr {Inner}\, \mc C_{n,L}$, whose underlying graded vector space agrees with that of
\[
\mr C^\bullet(\so(n), \sym^{>0}(\so(n)^*[-2]) \otimes \Omega^\bullet(\RR^n)) \otimes \mr C^\bullet(L)[n],
\]
with differential as in~\eqref{fib_dif}.
Note that the only change is in the far right term: $\mr C^\bullet_{\mr{red}}(L)[n]$ is replaced by $\mr C^\bullet(L)[n]$.
Such a change does not affect the proof above, which focuses on the $\so(n)$ contributions.
\end{proof}
 
\begin{remark}
Our main theorem establishes a sufficient condition for the vanishing of the framing anomaly of a topological AKSZ theory, namely the triviality of the cohomology group $\mr H^i(\so(n)) \otimes \mr H^{n-i}_{\mr{red}}(L)$ for all $i > 0$.   We believe that this condition will also be necessary.  In a follow-up paper \cite{EGWfr} we will establish the non-vanishing of the framing anomaly at the one-loop level in the case where this cohomology is non-trivial.  This is possible by evaluating the appropriate one-loop Feynman diagrams using a maximally holomorphic gauge fixing condition, by applying the results of \cite{BWhol}.  In this way it is possible to obtain a concrete identification of the one-loop framing anomaly in terms of characteristic classes.
\end{remark}

\printbibliography

\textsc{University of Massachusetts, Amherst}\\
\textsc{Department of Mathematics and Statistics, 710 N Pleasant St, Amherst, MA 01003}\\
\texttt{celliott@math.umass.edu}\\
\texttt{gwilliam@math.umass.edu}

\end{document}